\documentclass[runningheads]{llncs}

\usepackage{amsmath}
\usepackage{latexsym}
\usepackage{amsfonts,amssymb,amsmath}
\usepackage{graphicx, gastex}
\usepackage{cite}
\usepackage{mathrsfs}

\DeclareSymbolFont{rsfscript}{OMS}{rsfs}{m}{n}
\DeclareSymbolFontAlphabet{\mathrsfs}{rsfscript}
\DeclareMathOperator{\Mod}{ mod}
\DeclareMathOperator{\dt}{.}
\DeclareMathOperator{\rt}{rt}

\begin{document}
\title{Reset thresholds of automata with\\ two cycle lengths}
\author{Vladimir V. Gusev, Elena V. Pribavkina}

\institute{Institute of Mathematics and Computer Science,\\ Ural Federal University, Ekaterinburg, Russia\\
\email{vl.gusev@gmail.com, elena.pribavkina@gmail.com}}

\maketitle

\begin{abstract}
We present several series of synchronizing automata with multiple parameters,
generalizing previously known results.
Let $p$ and $q$ be two arbitrary co-prime positive integers, $q > p$.
We describe reset thresholds of the colorings
of primitive digraphs with exactly one cycle of length $p$ and one cycle of length $q$.
Also, we study reset thresholds of the colorings
of primitive digraphs with exactly one cycle of length $q$ and two cycles of length $p$.
%
%In the present paper we generalize several series of 
%slowly synchronizing automata presented in~\cite{AGV13}.
%Namely, $\mathrsfs{W}_n$, $\mathrsfs{D}'_n$ and $\mathrsfs{D}''_n$.
\end{abstract}

\section{Introduction}
%This paper was inspired by the article \cite{AGV13} by Ananichev, Gusev and Volkov.
%Here we generalize some of the results.
A \emph{complete deterministic} \emph{finite automaton} $\mathrsfs{A}$, or simply \emph{automaton}, is
a triple $\langle Q,\Sigma,\delta\rangle$, where $Q$ is a finite \emph{set of states},
$\Sigma$ is a finite \emph{input alphabet}, and $\delta: Q\times \Sigma\mapsto Q$ is a totally 
defined \emph{transition function}. Following standard notation, by $\Sigma^*$ we mean the set 
of all finite words over the alphabet $\Sigma$, including the empty word $\varepsilon$. 
The function $\delta$ naturally
extends to the free monoid $\Sigma^{*}$; this extension is still
denoted by $\delta$. Thus, via $\delta$, every word
$w\in\Sigma^*$ acts on the set $Q$. For each $v\in\Sigma^*$ and each $q\in Q$ we write
$q\dt v$ instead of $\delta(q,v)$ and let $Q\dt v=\{q\dt v\mid q\in
Q\}$. 

An automaton $\mathrsfs{A}$ is called
\emph{synchronizing}, if there is a word $w\in\Sigma^*$ %such that %for every $q, q' \in Q$ 
%we have $q \dt w = q' \dt w$,
which brings all states of the automaton $\mathrsfs{A}$ to a particular one,
 i.e. $|Q\dt w|=1.$ 
Any such word $w$ is said to be a \emph{reset} (or \emph{synchronizing}) \emph{word}
for the automaton $\mathrsfs{A}$. The minimum length of reset words for $\mathrsfs{A}$
is called the \emph{reset threshold} of $\mathrsfs{A}$. 
%This notion dates back to the sixties.
%Back then automata served as a mathematical models of devices working in discrete time.
%such as computers or relay control systems. This led to the following 
%natural problem:
%how can we restore control over such a device, if we do not know its current state? Reset word
%answers this question, since after applying it the state of the device becomes known.
%Nowadays 

%THEY ARE USEFULL
Synchronizing automata serve as transparent and natural models of error-resistant
systems in many applied areas (robotics, coding theory).
At the same time, synchronizing automata surprisingly arise in some parts of pure mathematics
(algebra, symbolic dynamics, combinatorics on words). See recent surveys
by Sandberg \cite{Sa05} and Volkov \cite{Vo_Survey}
for more details on the theory and applications of synchronizing automata. %The applications
%of synchronizing automata in coding theory and connections with symbolic dynamics
%are elegantly described in the recent book \cite{BePeReu09}
%by Berstel, Perrin and Reutenauer.
%BIG QUESTION

One of the most important and natural questions related to synchronizing automata is the following:
given $n$, how big can the reset threshold of an automaton with $n$ states be?
In 1964 \v{C}ern\'{y} exhibited a series of automata with $n$ states whose reset threshold equals 
$(n-1)^2$ \cite{Ce64}. Soon after he conjectured, that this series represents the worst possible case,
i.e. the reset threshold of every $n$-state synchronizing automaton is at most $(n-1)^2$.
This hypothesis has become known as
the \emph{\v{C}ern\'{y} conjecture}. In spite of its simple
formulation and many researchers' efforts, the \v{C}ern\'{y}
conjecture remains unresolved for about fifty years. 
Moreover, no upper bound of magnitude $O(n^2)$ for the 
reset threshold
of a synchronizing $n$-state automaton is known so far. The best known 
upper bound on the reset threshold of a synchronizing
$n$-state automaton is the bound $\frac{n^3-n}6$ found by
Pin~\cite{Pi83} in 1983. 

%В попытке получить новые идеи для разрешения гипотезы Черни и материал
%для проверки гипотез в работах() были представлены так называемые медленно синхронизируемые
%автоматы.
In an attempt to understand why the \v{C}ern\'{y} conjecture is so difficult to
resolve, researchers started to look for \emph{slowly synchronizing automata},
i.e. automata with $n$ states and reset threshold close to $(n - 1)^2$.
First series of such automata were presented in~\cite{AVZ}.
The number of known series of slowly synchronizing automata was significantly increased
in~\cite{AGV13}. In the latter paper the constructions are based on
the observed connection between slowly synchronizing
automata and primitive digraphs with large exponent.
%The latter advance is based on the connection between slowly synchronizing
%automata and directed graphs with large exponent.
%Discovered series were obtained from digraphs with large exponent.

A digraph $D$ is said to be \emph{primitive}, if there is a positive integer $t$ such that for every pair
of vertices $u$ and $v$ there is a path form $u$ to $v$ of length $t$. The smallest $t$ with this property 
is called the \emph{exponent} of the digraph $D$. Equivalently, if $M$ is the adjacency matrix of $D$,
then $t$ is the smallest number such that $M^t$ is positive. For additional results on the well-established field
of primitive digraphs we refer a reader to~\cite{BrRy91}.

The \emph{underlying digraph} $\mathcal{D}(\mathrsfs{A})$ of an automaton $\mathrsfs{A}$ has $Q$
as the set of vertices, and $(u,v)$ is an edge if $u \dt x = v$ for some letter $x \in \Sigma$.
A \emph{coloring} of a digraph $D$ is an automaton $\mathrsfs{A}$ such that $\mathcal{D}(\mathrsfs{A})$
is isomorphic to $D$.
%COLORING, ROAD COLORING.
Proposition~2~\cite{AGV13} states, that
the reset threshold of an arbitrary
$n$-state strongly connected synchronizing automaton is greater 
than the exponent of the underlying digraph minus $n$.
At the same time, the Road Coloring theorem~\cite{Tr09} states that any primitive digraph
has at least one synchronizing coloring.
Thus, $n$-state slowly synchronizing automata can be constructed from the well-known examples~\cite{DM64} of 
primitive digraphs on $n$ vertices with exponents close $(n-1)^2$. This idea was presented and
explored in~\cite{AGV13}. In the present paper we generalize several series of 
slowly synchronizing automata presented in~\cite{AGV13}.
Namely, $\mathrsfs{W}_n$, $\mathrsfs{D}'_n$ and $\mathrsfs{D}''_n$.

%On one hand we generalize presented examples. Namely, $\mathrsfs{W}_n$, $\mathrsfs{D}'_n$, $\mathrsfs{D}''_n$ from
%the paper~\cite{AGV13}. Obtain better understanding.
%interesting connection between
%of strongly connected synchronizing automata and \emph{exponents} of \emph{primitive digraphs}
%was discovered.

Another motivation for the present paper comes from the following facts.
Computational experiments of Trahtman~\cite{Tr06a} revealed that
not every positive integer in $\{1, \ldots, (n-1)^2\}$ may serve as the reset threshold of some automaton
with $n$ states over a binary alphabet. For example, there is no automaton with nine states over 
a binary alphabet with the reset threshold in the range from 59 to 63. Similar gaps were found for
automata with the number of states ranging from 6 to 10. These results were confirmed in~\cite{AGV13}.
Moreover, a second gap was presented, i.e. there are no 9-state automata over a binary alphabet with
the reset threshold from 53 to 55.
For 10-state automata a third gap, along with the first two, was found in the course of computational experiments of
Kisielewicz and Szyku{\l}a~\cite{KS13}.
This brings up the following
natural question: given $n$, which positive integers are
reset thresholds of $n$-state automata?
%In the present paper we approach the following 
Surprisingly, the set $E_n$ of all possible exponents of primitive digraphs on a fixed
number $n$ of vertices has similar gaps~\cite{DM64} as the set $R_n$ of all possible 
reset thresholds of $n$-state automata. Furthermore, for every $n$ the set $E_n$ is 
fully described~\cite[p. 83]{BrRy91}. We hope that study of this similarity
could shed light on properties of $R_n$. The following
statement~\cite{LV81} plays the key role in the description of $E_n$: if the exponent of a primitive digraph $D$
is at least $\frac{(n - 1)^2 + 1}{2} + 2$, then $D$ has cycles of exactly two different lengths.
%In order to describe $R_n$ one could use similarities between $R_n$ and $E_n$.
This motivates our choice in the present paper to focus on automata whose underlying digraphs have exactly 
two different cycle lengths. 

Let $p$ and $q$ be two arbitrary co-prime positive integers, $q > p$.
In section~\ref{sec:wielandt} we describe reset thresholds of the colorings
of primitive digraphs with exactly one cycle of length $p$ and one cycle of length $q$.
In section~\ref{sec:second} we study reset thresholds of the colorings
of primitive digraphs with exactly one cycle of length $q$ and two cycles of length $p$.

%The same question about exponents 
%On the other hand we aim at describing possible reset thresholds. Namely, 
%Trahtman~\cite{Tr06a} and my result. And \cite{KS13}. There are gaps. Similar behavior for exponents. See~\cite{DM64}.
%Total description. Major role -- digraphs with two cycle lengths.
%They fill in with a bunch of examples. We have similar goal. We color
%these digraphs(two cycles) in order to obtain variety of reset thresholds.
%
%Interesting examples, almost same behavior, like graph synchronization,
%number theory methods.Interesting case with only one or two
%division. Between automata and graphs.
%
%Synchronizing, applications,
%Cerny conjecture.
%Connection with large exponent graphs.
%We extend this results. Two cycles are
%important. They were used in description
%of possible exponents. We move into
%description of possible reset thresholds.
%Interesting case with only one or two
%division. Between automata and graphs.

\section{Wielandt-type automata}
\label{sec:wielandt}
%\begin{proposition}
%\label{lower bound} Let $\mathrsfs{A}=\langle
%Q,\Sigma,\delta\rangle$ be a \scn\ synchronizing $n$-automaton. Then
%\begin{equation}
%\label{eq:lower bound} \rt(\mathrsfs{A})\ge\gamma(D(\mathrsfs{A}))-n+1.
%\end{equation}
%\end{proposition}
%Thus in order to estimate the reset threshold of an automaton we need to find the exponent of its graph.
%We use the following well-known fact
We start with recalling the following elementary and well-known number-theoretic result.
\begin{theorem}[{\mdseries\cite[Theorem 2.1.1]{RaAl05}}]
\label{coin_problem}
Given two positive co-prime integers $p$ and $q$, the largest integer that is not expressible as a non-negative integer combination
of $p$ and $q$, is $(p-1)(q-1)-1$.
\end{theorem}

Let us fix two positive co-prime integers $p$ and $q$. Without loss of generality, we assume $p<q$. 
Let $n$ be a positive integer, $n<p+q$. We define a Wielandt-type automaton $\mathrsfs{W}(n,q,p)$ 
as follows (see Fig.~\ref{fig_A1m}). 
%We start with considering Wielandt-type automata over the alphabet $\Sigma$ with $n<p+q$ states
%which have exactly one cycle of length $p$ and exactly one cycle of length $q$.
%whose cycles have lengths $p$ and $q$ (see Fig.~\ref{fig_A1m}). 
The state set $Q=\{0,1,\ldots, n-1\}$, $\Sigma=\{a,b\}$,
and the transitions are defined in the following way:\\ 
$0\dt a=q$ if $n>q$, and $0\dt a=q-p+1$ if $n=q$;\quad $0\dt b=1$; \\ 
$i\dt x=i+1$ for $1\le i< n-1$ and $i\ne q-1$ for each $x\in\Sigma$;\\ 
$(q-1)\dt x =0$ for each $x\in\Sigma$;\\
if $n>q$, then $(n-1)\dt x=n-p+1$ for each $x\in\Sigma$.

\smallskip
%The graphs of such automata are isomorphic to the one shown on Fig.~\ref{fig_A1mg}.
%This automaton will be referred to as Wielandt-type
%$(n,q,p)$-automaton. 
\noindent In case $q=n$, $p=n-1$ we obtain Wielandt automaton $\mathrsfs{W}_n$ considered in \cite{AGV13}. %, where $n$ is the number of states. %superscript $^{(1)}$ stands for the number of cycles of each length.
%Its underlying graph will be referred to as $(n,q^{(1)},p^{(1)})$-graph.
It is not hard to observe, that every strongly connected $n$-state automaton whose underlying digraph has exactly one cycle of length $p$
and exactly one cycle of length $q$ is isomorphic to $\mathrsfs{W}(n,q,p)$.

First let us consider the case $n=q$ (see Fig.~\ref{fig_A1}).
\begin{figure}[ht]
 \begin{center}
  \unitlength=4pt
    \begin{picture}(30,30)(0,0)
    \gasset{Nw=4,Nh=4,Nmr=2}
    \thinlines
%    \node(B0)(0,15){$XX$}
%    \node(B1)(30,15){$XY$}
    \node[Nw=8,Nh=8,Nmr=4,Nframe=n](A6)(15,0){$\ldots$}
\node(A0)(1,20){$0$}
\node(A1)(7.5,28){$1$}
\node(A2)(22.5,28){$2$}
\node[Nw=8,Nh=8,Nmr=4,Nframe=n](A3)(29,20){$\vdots$}
\node[Nw=6](A4)(4.4,4.4){$q$-$1$}
\node[Nw=7](A5)(28,7.5){$q$-$p$+$1$}
%    \node[Nw=6](A5)(22,3){$q-1$}
%    \node(A6)(10,2){$0$}
    \drawedge[curvedepth=1](A0,A1){$b$}
    \drawedge[curvedepth=2](A1,A2){$a,b$}
    \drawedge[curvedepth=1](A2,A3){$a,b$}
    \drawedge[curvedepth=2](A4,A0){$a,b$}
    \drawedge[curvedepth=1](A3,A5){}
    \drawedge[curvedepth=1.5](A5,A6){$a,b$}
    \drawedge[curvedepth=1](A6,A4){}
    \drawedge[curvedepth=2](A0,A5){$a$}
    \end{picture}
 \end{center}
  \caption{The Wielandt-type automaton $\mathrsfs{W}(q,q,p)$}
  \label{fig_A1}
\end{figure}
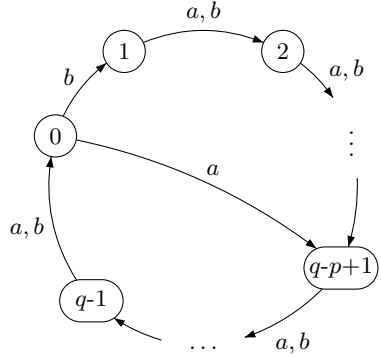
\begin{lemma}
\label{rt}
Let $\mathrsfs{A}$ be a strongly connected synchronizing automaton, whose cycles have lengths $p$ and $q$. If $\gcd(p,q)=1$, then $\rt(\mathrsfs{A})\ge (p-1)(q-1)$.
Moreover, if there are states $s$, $t$, and a positive integer $\ell$ such that:\\
(i) there is a shortest synchronizing word $w$ which resets the automaton $\mathrsfs{A}$ to $s$, \\
(ii) $t \dt u = s$ for each word $u$ of length $\ell$,\\
then $\rt(\mathrsfs{A})\ge (p-1)(q-1) + \ell$.
\end{lemma}
\begin{proof} Let $\mathrsfs{A}=\langle Q,\Sigma,\delta\rangle$.
We prove the first part of the lemma.
Consider a synchronizing word $w$ having shortest possible length. Let $s=Q\dt w$ be the state to which the automaton is synchronized. Note, that the word $uw$ is synchronizing
for every $u\in\Sigma^*$, and $Q\dt uw=s.$ In particular, we have $s\dt w=s\dt uw=s$. Thus the word $w$, as well as the word $uw$, for every word $u$, labels a path in the automaton $\mathrsfs{A}$ from the state $s$ to itself. Every such path can be decomposed into cycles of lengths $p$ and $q$. Hence 
the number $|w|$, as well as $|w|+k$, for each positive integer $k$, can be represented as a non-negative combination of the numbers $p$ and $q$. Thus, by theorem~\ref{coin_problem},
we have $\rt(\mathrsfs{A})\ge (p-1)(q-1).$

Assume now that in addition there exist a state $t$ and a positive integer $\ell$
such that $t \dt u = s$ for each word $u$ of length $\ell$.
Suppose, contrary to our claim, that $|w| < (p-1)(q-1) + \ell$. Let $u \in \Sigma^*$ be an arbitrary word such that
$|uw| = (p-1)(q-1) + \ell - 1$. As before, the word $uw$ synchronizes the automaton $\mathrsfs{A}$
to the state $s$. But after applying its prefix of length $\ell$ to the state $t$ we end up in
the state $s$. Hence there is a path of length $(p-1)(q-1) - 1$ from $s$ to itself. But this number can not
be represented as a non-negative combination of $p$ and $q$ by theorem~\ref{coin_problem}. A~contradiction.
\end{proof}
\begin{theorem}
\label{rt_A1} The reset threshold of the Wielandt-type automaton $\mathrsfs{W}(q,q,p)$ equals $(p-1)(q-1)+q-p.$
\end{theorem}
\begin{proof}
Any shortest reset word $w$ for this automaton resets it to the state $q-p+1$, since it
is the only state which is a common end of two different edges with the same label.
Note, that any word of length $q-p$ brings the state 1 to the state $q-p+1$.
Lemma~\ref{rt} implies that the reset threshold of $\mathrsfs{W}(q,q,p)$ is at least $(p-1)(q-1) + q - p$.
%\emph{Fact 1.} There is no word $u$ of length $(p-1)(q-1)-1$ such that $(q-p+1)\dt u=q-p+1.$
%
%\noindent Indeed, any word $u$ such that $(q-p+1)\dt u=q-p+1$ forms a directed cycle consisting of
%simple cycles of lengths $p$ and $q$. Thus, the length of $u$ must be  a non-negative integer combination of the numbers
%$p$ and $q$, but $(p-1)(q-1)-1$ is not by the theorem~\ref{coin_problem}.
%
% Assume, that $|w|=(p-1)(q-1)$. Thus, there is a directed path from the state $q-p$ to
%the state $q-p+1$ of length $(p-1)(q-1)$. Since both transitions from $q-p$ lead to $q-p+1$, we conclude,
%that there is a directed cycle on the state $q-p+1$ of length $(p-1)(q-1)-1,$ which contradicts Fact~1.
%Thus, $|w|\ge(p-1)(q-1)+1$. Suppose, that $|w|\ge (p-1)(q-1)+i$ for $1\le i< q-p$, and let us prove, that
%$|w|\ge(p-1)(q-1)+i+1$. Indeed, if $|w|=(p-1)(q-1)+i$, then we consider the action of $w$ on the state $q-p-i$.
%From $q-p-i\ge 1$ it follows that after applying the first $i+1$ letters of $w$ we end up in the state $q-p+1$.
%After applying the remaining $(p-1)(q-1)-1$ letters we should return to $q-p+1$, but this is impossible by the Fact~1.
%Thus, we have $\rt(\mathrsfs{A})\ge (p-1)(q-1)+q-p$.

Let us check that the word $w=a^{q-p}(ba^{q-1})^{p-2}ba^{q-p}$ synchronizes $\mathrsfs{W}(q,q,p)$.
After applying the prefix $a^{q-p}$ we end up in the cycle $C$ of length $p$:
$$Q\dt a^{q-p}=\{0,q-p+1,q-p+2,\ldots,q-1\}.$$ Next, we show that that the word $(ba^{q-1})^{p-2}$
brings $C$ to a two-element set. We state this fact as a separate lemma:

\begin{lemma}
\label{Cycle_syn} Let $\mathrsfs{A}$ be an automaton with the state set $Q$ over the alphabet $\Sigma=\{a,b\}$. Let $q > p$ be two co-prime
positive integers, and let $r$ denote the remainder of the division of $q$ by $p$.
Let $C = \{0,1, \ldots, p - 1\}$ be
a subset of $Q$ such that $0\dt a = 1$, $0\dt ba^{q-1} = 0$, and $i \dt x \equiv i+1 \Mod p$ for $1\leq i \leq p - 1$
and for all $x \in \Sigma$. Then $C \dt (ba^{q-1})^{p-2} = \{0, p - r\}$.
\end{lemma}
\begin{proof}
First note, that $i \dt ba^{q-1} \equiv i+r \Mod p$ for each state $i \neq 0$.
Consider the equation $i + rx \equiv 0 \Mod p$. Since $r$ and $p$ are co-prime,
this equation has unique solution in $\{1, \ldots, p - 1\}$ for every $i \neq 0$.
Then $i \dt (ba^{q-1})^x = 0$.
If $x \neq p - 1$, then $i \dt (ba^{q-1})^{p - 2} = 0$.
The case $x = p - 1$ occurs only if $i = r$. In this case $r \dt (ba^{q-1})^{p - 2} = p - r$.
%The path labeled by the factor $ba^{q-1}$ starting from each state
%of this cycle except $0$ passes only through states on this cycle. As $p$ and $q$ are co-prime, this means, that
%every time we apply the factor $ba^{q-1}$ to a state in $\{q-p+1,q-p+2,\ldots,q-1\}$ we either end up in
%a different state on this cycle, or in $0$. In the latter case $0\dt ba^{q-1}=0$. Since there are $p$ states
%on the cycle, after applying the factor $ba^{q-1}$ $p-2$ times only two states will remain: $0$ and a state $r$ such
%that $r\dt ba^{q-1}=0$. But then it is easy to see, that $0\dt ba^{q-p}=r\dt ba^{q-p}=q-p+1.$
\end{proof}
Returning back to the proof of the theorem, we have $C\dt (ba^{q-1})^{p-2}=\{0, q-r\}$.
The word $ba^{q-p}$ brings the latter set to the singleton $q-p+1$.
\end{proof}

Let us consider now the general case of the Wielandt-type automaton $\mathrsfs{W}(n,q,p)$ (see Fig.~\ref{fig_A1m}).
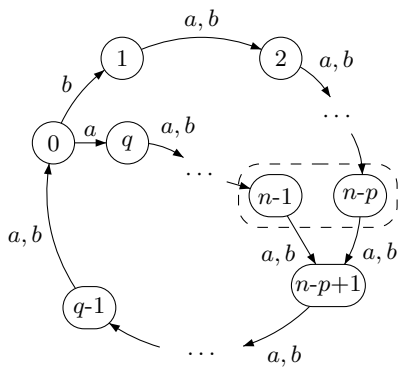
\begin{figure}[ht]
 \begin{center}
  \unitlength=4pt
    \begin{picture}(30,30)(0,0)
    \gasset{Nw=4,Nh=4,Nmr=2}
    \thinlines
    \node[Nw=8,Nh=8,Nmr=4,Nframe=n](A6)(15,0){$\ldots$}
\node(A0)(1,20){$0$}
\node(A1)(7.5,28){$1$}
\node(A2)(22.5,28){$2$}
\node[Nframe=n](A3)(28,22.5){$\cdots$}
\node[Nw=5](A4)(4.4,4.4){$q$-$1$}
\node[Nw=7](A5)(27,6.5){$n$-$p$+$1$}
\node[Nw=5,Nframe=n](B4)(15,17){$\cdots$}
\node(B1)(8,20){$q$}
\node[Nw=5](B2)(22,15){$n$-$1$}
\node[Nw=5](B3)(30,15){$n$-$p$}
\node[dash={1}0, Nw=15,Nh=6](C)(26,15){}
    \drawedge[curvedepth=1](A0,A1){$b$}
    \drawedge[curvedepth=2](A1,A2){$a,b$}
    \drawedge[curvedepth=1](A2,A3){$a,b$}
    \drawedge[curvedepth=2](A4,A0){$a,b$}
    \drawedge[curvedepth=1](A3,B3){}
    \drawedge[curvedepth=1.5](A5,A6){$a,b$}
    \drawedge[curvedepth=1](A6,A4){}
    \drawedge(A0,B1){$a$}
    \drawedge[curvedepth=1](B1,B4){$a,b$}
    \drawedge(B4,B2){}
    \drawedge[ELside=r](B2,A5){$a,b$}
    \drawedge[curvedepth=1](B3,A5){$a,b$}
    \end{picture}
 \end{center}
  \caption{The Wielandt-type automaton $\mathrsfs{W}(n,q,p)$}
  \label{fig_A1m}
\end{figure}
%\begin{figure}[ht]
% \begin{center}
%  \unitlength=4pt
%    \begin{picture}(30,30)(0,0)
%    \gasset{Nw=4,Nh=4,Nmr=2}
%    \thinlines
%    \node[Nw=8,Nh=8,Nmr=4,Nframe=n](A6)(15,0){$\ldots$}
%\node(A0)(1,20){$0$}
%\node(A1)(7.5,28){$1$}
%\node(A2)(22.5,28){$2$}
%\node[Nframe=n](A3)(28,22.5){$\cdots$}
%\node[Nw=5](A4)(4.4,4.4){$q$-$1$}
%\node[Nw=7](A5)(27,6.5){$n$-$p$+$1$}
%\node[Nw=5,Nframe=n](B4)(15,17){$\cdots$}
%\node(B1)(8,20){$q$}
%\node[Nw=5](B2)(22,15){$n$-$1$}
%\node[Nw=5](B3)(30,15){$n$-$p$}
%%\node[dash={1}0, Nw=15,Nh=6](C)(26,15){}
%    \drawedge[curvedepth=1](A0,A1){}
%    \drawedge[curvedepth=2](A1,A2){}
%    \drawedge[curvedepth=1](A2,A3){}
%    \drawedge[curvedepth=2](A4,A0){}
%    \drawedge[curvedepth=1](A3,B3){}
%    \drawedge[curvedepth=1.5](A5,A6){}
%    \drawedge[curvedepth=1](A6,A4){}
%    \drawedge(A0,B1){}
%    \drawedge[curvedepth=1](B1,B4){}
%    \drawedge(B4,B2){}
%    \drawedge[ELside=r](B2,A5){}
%    \drawedge[curvedepth=1](B3,A5){}
%    \end{picture}
% \end{center}
%  \caption{$(n,q^{(1)},p^{(1)})$-graph}
%  \label{fig_A1mg}
%\end{figure}
It is rather easy to see, that given a synchronizing automaton $\mathrsfs{B}$ and a congruence $\rho$,
the factor automaton $\mathrsfs{B}/\rho$ is also synchronizing, and $\rt(\mathrsfs{B}/\rho)\le \rt(\mathrsfs{B})$.
In particular, consider the following congruence $\sigma$ on $\mathrsfs{B}$: for two states $s$ and $t$ we have
$s\sigma t$ if and only if $s\dt x=t\dt x$ for each $x\in\Sigma.$
\begin{lemma}\label{congr} If $\mathrsfs{B}$ is synchronizing, then $\mathrsfs{B}/\sigma$ is also synchronizing, and
$$\rt(\mathrsfs{B}/\sigma)\le \rt(\mathrsfs{B})\le \rt(\mathrsfs{B}/\sigma)+1.$$
\end{lemma}
\begin{proof}
The inequality $\rt(\mathrsfs{B}/\rho)\le \rt(\mathrsfs{B})$ is trivial. 
The states of $\mathrsfs{B}/\sigma$ are congruence classes $[s]^\sigma$ of the states $s$ of the automaton $\mathrsfs{B}$.
Let us consider a synchronizing word $w$ for the automaton $\mathrsfs{B}/\sigma$. %, which resets it to the state $[t]^\sigma$.
For every pair of states $s$ and $s'$ of the original automaton $\mathrsfs{B}$ we have $s\dt w\,\sigma\, s'\dt w.$
But this means that $s\dt wx=s'\dt wx$ for any letter $x\in\Sigma$, thus, the word $wx$ resets the automaton $\mathrsfs{B}$.
Thus we have $\rt(\mathrsfs{B})\le \rt(\mathrsfs{B}/\sigma)+1$.
\end{proof}
\begin{lemma}\label{cong_A} If $n>q$, then $\mathrsfs{W}(n,q,p)/\sigma$ is equal to $\mathrsfs{W}(n-1,q,p)$, and $$\rt(\mathrsfs{W}(n,q,p))=\rt(\mathrsfs{W}(n-1,q,p))+1.$$
\end{lemma}
\begin{proof}
Let $w$ be a word of minimal length, synchronizing the automaton $\mathrsfs{W}(n,q,p)$. As in the proof of theorem~\ref{rt_A1}, the word $w$ resets $\mathrsfs{W}(n,q,p)$ to the state $n-p+1$. On the last step $w$ brings the states $\{n-1,n-p\}$ to the state $n-p+1$. Hence $w=w'x$, where $x\in\Sigma$, and $w'$ brings the automaton $\mathrsfs{W}(n,q,p)$ to the set $\{n-1,n-p\}$. But
these two states form the unique non-trivial $\sigma$-class (see Fig.~\ref{fig_A1m}). Thus the factor automaton $\mathrsfs{W}(n,q,p)/\sigma$ is equal to the Wielandt-type automaton $\mathrsfs{W}(n-1,q,p)$. Moreover, it is synchronized by $w'$. Thus, $\rt(\mathrsfs{W}(n-1,q,p))\le \rt(\mathrsfs{W}(n,q,p))-1$.
On the other hand, by lemma~\ref{congr} we have $\rt(\mathrsfs{W}(n-1,q,p))\ge \rt(\mathrsfs{W}(n,q,p))-1$. Therefore, we get the required equality.  %only to this class from the singleton classes $[n-2]^\sigma$ ($[0]^\sigma$ in case $n=q+1$) and $[n-p-1]^\sigma$.
\end{proof}
\begin{theorem}
\label{wielandt_general}
The reset threshold of the Wielandt-type automaton $\mathrsfs{W}(n,q,p)$ is equal to $(p-1)(q-1)+n-p.$
\end{theorem}
\begin{proof}
Since there are $n-q$ states on the path from the state $0$ to $n-p+1$, lemma~\ref{cong_A} can be applied $n-q$ times
to obtain the Wielandt-type automaton $\mathrsfs{W}(q,q,p)$. By theorem~\ref{rt_A1}, its reset threshold equals $(p-1)(q-1)+q-p$. 
Each time lemma~\ref{cong_A}
is applied, the reset threshold is decreased strictly by 1. Thus the reset threshold of the automaton $\mathrsfs{W}(n,q,p)$
is equal to $(p-1)(q-1)+n-p$.
\end{proof}

\section{Dulmage-Mendelsohn-type automata}
\label{sec:second}
As in the previous section, let $q$ and $p$ be two co-prime positive integers, and $q>p$. Let $k$ be a positive integer such that $k<\min\{p,q-p+1\}$.
Here we consider Dulmage-Mendelsohn-type automata, which are the colorings of the following primitive digraph $D(q,p,k)$ (see Fig.~\ref{fig_B}).
Its vertex set is $\{0,\ldots,q-1\}$, the set of edges is $\{(i,(i+1)\Mod q)\mid 0\le i< q\}\cup\{(0,q-p+1),(k,(q-p+k+1)\Mod q)\}.$
Note, that $D(q,p,k)$ has exactly one cycle of length $q$ and two cycles of length $p$.
The digraph $D(q,p,k)$ has only two non-isomorphic colorings $\mathrsfs{D}^{aa}(q,p,k)$ and $\mathrsfs{D}^{ab}(q,p,k)$ (see Fig.~\ref{fig_B1}). 
%Consider now strongly connected synchronizing automata having one cycle of length $q$ and two cycles of length $p$. First we consider such automata
%with exactly $q$ states. Let $k$ be the minimal distance on the $q$-cycle between two states $s,s'$ such that $s\dt a\ne s\dt b$ and $s'\dt a\ne %s'\dt b$.
%Obviously, $1\le k\le q-p$.
%Assume for now that $k<p$. The parameters $p$, $q$, and $k$ define up to isomorphism the graph with one $q$-cycle and
%two $p$-cycles, see Fig.~\ref{fig_B1}. Such graphs are referred to as $(q,q^{(1)},p^{(2)},k)$-graphs.
%There are two non-isomorphic colorings of a $(q,q^{(1)},p^{(2)},k)$-graph: the one in which edges from $0$ and $k$ are labeled by the same letter,
% and the one in which they are labeled by different letters. The first coloring will be denoted by $\mathrsfs{B}_1$, and the second one
% by $\mathrsfs{B}_2.$
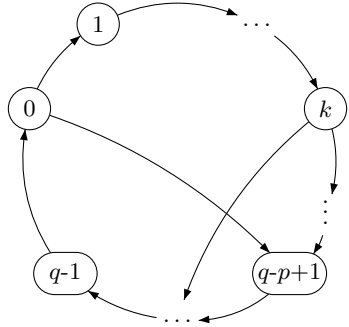
\begin{figure}[ht]
 \begin{center}
  \unitlength=4pt
    \begin{picture}(30,25)(0,2)
    \gasset{Nw=4,Nh=4,Nmr=2}
    \thinlines
    \node[Nframe=n](A6)(15,0){$\ldots$}
\node(A0)(1,20){$0$}
\node(A1)(7.5,28){$1$}
\node[Nframe=n](A2)(22.5,28){$\ldots$}
\node(A3)(29,20){$k$}
\node[Nw=6](A4)(4.4,4.4){$q$-$1$}
\node[Nw=7](A5)(25.6,4.4){$q$-$p$+$1$}
\node[Nframe=n,NLdist=1](A7)(29,10){$\vdots$}
    \drawedge[curvedepth=1](A0,A1){}
    \drawedge[curvedepth=2](A1,A2){}
    \drawedge[curvedepth=1](A2,A3){}
    \drawedge[curvedepth=2](A4,A0){}
    \drawedge[curvedepth=1](A3,A7){}
    \drawedge[curvedepth=1](A7,A5){}
    \drawedge[curvedepth=1.5](A5,A6){}
    \drawedge[curvedepth=1](A6,A4){}
    \drawedge[curvedepth=2](A0,A5){}
    \drawedge[curvedepth=-2](A3,A6){}
    \end{picture}
 \end{center}
  \caption{Digraph $D(q,p,k)$}
  \label{fig_B}
\end{figure}
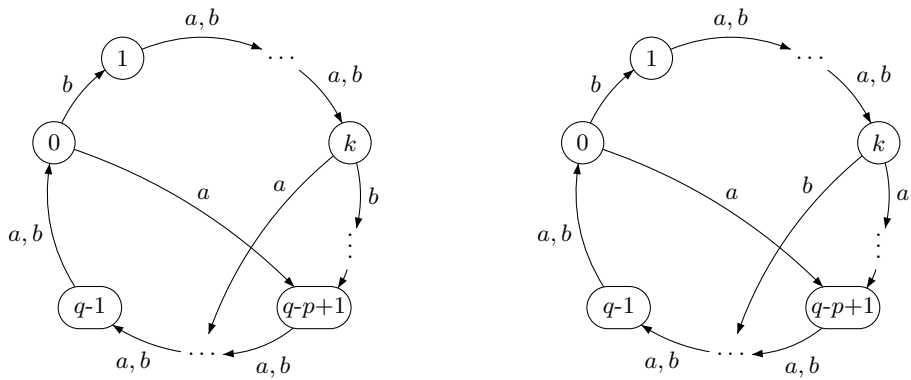
\begin{figure}[ht]
 \begin{center}
  \unitlength=4pt
    \begin{picture}(60,35)(50,0)
    \gasset{Nw=4,Nh=4,Nmr=2}
    \thinlines

\node[Nframe=n](B6)(55,0){$\ldots$}
\node[Nframe=n,NLdist=1](B7)(69,10){$\vdots$}
\node(B0)(41,20){$0$}
\node(B1)(47.5,28){$1$}
\node[Nframe=n](B2)(62.5,28){$\ldots$}
\node(B3)(69,20){$k$}
\node[Nw=6](B4)(44.4,4.4){$q$-$1$}
\node[Nw=7](B5)(65.6,4.4){$q$-$p$+$1$}
    \drawedge[curvedepth=1](B0,B1){$b$}
    \drawedge[curvedepth=2](B1,B2){$a,b$}
    \drawedge[curvedepth=1](B2,B3){$a,b$}
    \drawedge[curvedepth=2](B4,B0){$a,b$}
    \drawedge[curvedepth=1](B3,B7){$b$}
    \drawedge[curvedepth=1](B7,B5){}
    \drawedge[curvedepth=1.5](B5,B6){$a,b$}
    \drawedge[curvedepth=1](B6,B4){$a,b$}
    \drawedge[curvedepth=2](B0,B5){$a$}
    \drawedge[curvedepth=-2,ELpos=30,ELside=r](B3,B6){$a$}

\node[Nframe=n](C6)(105,0){$\ldots$}
\node[Nframe=n,NLdist=1](C7)(119,10){$\vdots$}
\node(C0)(91,20){$0$}
\node(C1)(97.5,28){$1$}
\node[Nframe=n](C2)(112.5,28){$\ldots$}
\node(C3)(119,20){$k$}
\node[Nw=6](C4)(94.4,4.4){$q$-$1$}
\node[Nw=7](C5)(115.6,4.4){$q$-$p$+$1$}
    \drawedge[curvedepth=1](C0,C1){$b$}
    \drawedge[curvedepth=2](C1,C2){$a,b$}
    \drawedge[curvedepth=1](C2,C3){$a,b$}
    \drawedge[curvedepth=2](C4,C0){$a,b$}
    \drawedge[curvedepth=1](C3,C7){$a$}
    \drawedge[curvedepth=1](C7,C5){}
    \drawedge[curvedepth=1.5](C5,C6){$a,b$}
    \drawedge[curvedepth=1](C6,C4){$a,b$}
    \drawedge[curvedepth=2](C0,C5){$a$}
    \drawedge[curvedepth=-2,ELpos=30,ELside=r](C3,C6){$b$}

    \end{picture}
 \end{center}
  \caption{Two Dulmage-Mendelsohn-type automata $\mathrsfs{D}^{aa}(q,p,k)$ and $\mathrsfs{D}^{ab}(q,p,k)$}
  \label{fig_B1}
\end{figure}
\begin{lemma}
\label{B2_synstate}
\begin{enumerate}
\item[(i)] Any shortest synchronizing word of the automaton $\mathrsfs{D}^{ab}(q,p,k)$ synchronizes it to the state $q-p+1$.
\item[(ii)] Any shortest synchronizing word of the automaton $\mathrsfs{D}^{aa}(q,p,k)$ synchronizes it to the state $q-p+1$ when $k < q - p$.

\end{enumerate}
\end{lemma}

\begin{proof}
Part (i). Let $t=q-p+k+1$. %be the state of $\mathrsfs{D}^{ab}(q,p,k)$ such that 
Note, that $t=k\dt b=(q-p+k)\dt a=(q-p+k)\dt b$. 
Any shortest synchronizing word $w$
can synchronize the automaton $\mathrsfs{D}^{ab}(q,p,k)$ either to $q-p+1$ or $t$. Suppose, that $w$ synchronizes $\mathrsfs{D}^{ab}(q,p,k)$ to the state $t$. By lemma~\ref{rt} we have $|w|\ge (p-1)(q-1).$ Moreover, $(p-1)(q-1)>k.$
Consider the suffix $v$ of $w$ of length $k$. It is easy to see, that the full preimage $t\dt v^{-1}$ of the state $t$
under the action of the word $v$ is equal to $\{1,q-p+1\}.$
If $k=q-p$, then the two incoming edges to the state $q-p+1$ are labeled by the letter $a$, while the only incoming edge
to the state $1$ is labeled by the letter $b$. A contradiction. If $k\ne q-p$, then the set $\{1,q-p+1\}$
was necessarily obtained from the set $\{0,q-p\}$ by applying the letter $b$. But $\{0,q-p\}\dt a=q-p+1$. Therefore,
we can replace the suffix of $w$ of length $k+1$ by the letter $a$, in order to obtain a shorter synchronizing word.
A contradiction. Hence the word $w$ synchronizes the automaton $\mathrsfs{D}^{ab}(q,p,k)$ to the state $q-p+1$.

The proof of the part (ii) of the lemma is analogous to the part (i) with only minor changes.
\end{proof}

\begin{theorem}
\label{B2_rt}
The reset threshold of the Dulmage-Mendelsohn-type automaton $\mathrsfs{D}^{ab}(q,p,k)$ is equal to $(p-1)(q-1)+q-p-k.$
\end{theorem}
\begin{proof}
Let $w$ be a reset word for the automaton $\mathrsfs{D}^{ab}(q,p,k)$ having minimal length.
By lemma~\ref{B2_synstate} the word $w$ synchronizes the automaton to the state  $q-p+1$. 
Note, that any word of length
$q - p - k$ brings the state $k + 1$ to the state $q-p+1$. Lemma~\ref{rt}
implies $|w|\ge(p-1)(q-1) + q - p - k$.

First let us assume that $k=q-p$. In this case it remains to prove that the
word $w_1=(ba^{q-1})^{p-2}ba^{q-p}$ is synchronizing.
Let $C$ be the cycle $\{0,q-p+1,q-p+2,\ldots, q-1\}$. Note, that the word $ba^{q-1}$
maps all the states, that do not belong to $C$, to the set $C.ba^{q-1}$.
Namely, $k\dt ba^{q-1}=(t-1)\dt ba^{q-1}$,
where $t=k\dt b$; $(k-1)\dt ba^{q-1}=(q-1)\dt ba^{q-1},(k-2)\dt ba^{q-1}=(q-2)\dt ba^{q-1},\ldots, 1\dt ba^{q-1}=(q-k+1=p+1)\dt ba^{q-1}.$
%Since $k<p$, all the resulting states $t-1,q-1,\ldots,p+1$ are on the cycle
Thus it is enough to consider the action of the word $w_1$ on the cycle $C$.
By lemma~\ref{Cycle_syn} we have $C \dt (ba^{q-1})^{p-2} = \{0, q - r\}$, where $r$ is the remainder
of the division of $q$ by $p$.
%that the path labeled by the factor $ba^{q-1}$ starting from each state
%on the cycle $C$ except $0$ passes only through states on this cycle. As $p$ and $q$ are co-prime, this means, that
%every time we apply the factor $ba^{q-1}$ to a state in $C\setminus\{0\}$ we either end up in
%a different state on this cycle, or in $0$. In the latter case $0\dt ba^{q-1}=0$. Since there are $p$ states
%on the cycle, after applying the factor $ba^{q-1}$ $p-2$ times only two states will remain: $0$ and a state $r$ such
%that $r\dt ba^{q-1}=0$.
But then it is easy to see, that $0\dt ba^{q-p}=(q-r)\dt ba^{q-p}=q-p+1.$

Now assume that $k<q-p$.
%In this case the argument follows the proof of theorem~\ref{rt_A1}.
%First we note, that there is no word $u$ of length $(p-1)(q-1)-1$ such that $(q-p+1)\dt u=q-p+1.$
%Then we assume, that $|w|=(p-1)(q-1)$. Thus, there is a directed path from the state $q-p$ to
%the state $q-p+1$ of length $(p-1)(q-1)$. Since both transitions from $q-p$ lead to $q-p+1$, we conclude,
%that there is a directed cycle on the state $q-p+1$ of length $(p-1)(q-1)-1,$ a contradiction.
%Thus, $|w|\ge(p-1)(q-1)+1$. Suppose, that $|w|\ge (p-1)(q-1)+i$ for $1\le i< q-p-k$, and let us prove, that
%$|w|\ge(p-1)(q-1)+i+1$. Indeed, if $|w|=(p-1)(q-1)+i$, then we consider the action of $w$ on the state $q-p-i$.
%Since $1 \leq q-p-i \leq q-p-1$ it follows that after applying the first $i+1$ letters of $w$ we end up in the state $q-p+1$.
%After applying the remaining $(p-1)(q-1)-1$ letters we should return to $q-p+1$, but this is impossible.
%Thus, $|w|\ge(p-1)(q-1)+i+1$ for $1\le i< q-p-k$. Therefore $\rt(\mathrsfs{A})\ge (p-1)(q-1)+q-p-k$.
Let us show that the word $$w_2=ba^{q-p-k-1}(ba^{q-1})^{p-2}ba^{q-p}$$ is synchronizing.
All the states in the range from $k$ to $q-p$ are mapped into
the cycle $C$ under the action of the prefix $ba^{q-p-k-1}$.
This prefix maps the remaining states lying outside the cycle $C$, i.e. $1,2, \ldots, k - 1$,
to the states ranging from $q-p-k + 1$ to $q-p-1$.
Namely, $(k-i)\dt ba^{q-p-k-1}=q-p-i$ for $1\le i\le k-1$.
%$\{1,\ldots,k-1\}\dt ba^{q-p-k-1}=\{q-p-k+1,\ldots,q-p-1.\}$
The action of the word $ba^{q-1}$ on the states in $\{q-p-k + 1, \ldots, q-p-1\}$
coincides with the action of this word on some states in the cycle $C$.
More precisely, we have $(q-p-i)\dt ba^{q-1}=(q-i)\dt ba^{q-1}$ for  $1\le i\le k-1$,
provided that for no such $i$ we have $q-p-i=k$. If $q-p-i=k$ for some $i$,
then we have $k\dt ba^{q-1}=(t-1)\dt ba^{q-1}$, where $t=k\dt b$.
In both cases the condition $k<p$ implies that all the resulting states
$t-1,q-1,\ldots, q-k+1$ lie on the cycle $C$.
Hence the word $w_2$ brings the automaton $\mathrsfs{D}^{ab}(q,p,k)$ into the subset of $C\dt (ba^{q-1})^{p-2}ba^{q-p}$.
As we have already seen, the latter set is the singleton $q-p+1$.
\end{proof}

\begin{theorem}
\label{B1_rt} The reset threshold of the Dulmage-Mendelsohn-type automaton $\mathrsfs{D}^{aa}(q,p,k)$ equals $(p-1)(q-1)+q-p-k$ if $k<q-p$, and  $(p-1)(q-1)+2(q-p)$ if $k=q-p$.
\end{theorem}

\begin{proof}
First let us assume that $k<q-p$.
Let $w$ be reset word for the automaton $\mathrsfs{D}^{aa}(q,p,k)$ having minimal possible length. Lemma~\ref{B2_synstate} implies that
the word $w$ brings the automaton to the state $q-p+1$. Note, that any word of length $q-p-k$ brings the state $k+1$ to the state
$q-p+1$. Thus by lemma~\ref{rt} we have $|w|\ge (p-1)(q-1)+q-p-k$.

Let us prove that the word $w_1=a^{q-p-k}(ba^{k-1}ba^{q-k-1})^{p-2}ba^{k-1}ba^{q-p-k}$ is synchronizing.
Consider the cycle $C=\{0,q-p+1,q-p+2,\ldots,q-1\}$. Note, that the prefix $a^{q-p-k}$
maps the states, ranging from $k+1$ to $q-p$, to the states in $C$.
%since $k+i\dt a^{q-p-k}=q-p+i$ for $1\le i\le q-p-k$.
Consider now the action of the prefix $a^{q-p-k}$ on the states from $1$ to $k$.
If $q-p-k+1>k$, then all these states are mapped to some states in $C$.
If $q-p-k+1\le k$, then these states are mapped into $C\cup\{q-p-k+1,\ldots,k\}$.
Next, for each state $t$ from $q-p-k+1$ to $k$ we present a state $t'$ from $C$ such that
$t\dt ba^{k-1}=t'\dt ba^{k-1}$. If $t\ne k$, then it is easy to check that $t'=q-p+t$. Since
$q-p-k+1>1$, we have $t'>q-p+1$. Hence the state $t'\in C$.
If $t=k$, then $t'=k+p$ (recall, that $k+p<q$). The state $k+p$ belongs to $C$. Indeed,
from $q-p-k+1\le k$ and $k<p$ we obtain $k+p>2k\ge q-p+1$.
Hence the word $w_1$ brings the automaton $\mathrsfs{D}^{aa}(q,p,k)$ into the subset of $C\dt (ba^{k-1}ba^{q-k-1})^{p-2}ba^{k-1}ba^{q-p-k}$.
Thus it remains to show, that the latter set is a singleton.
The argument is similar to the proof of lemma~\ref{Cycle_syn}. Instead of the word $ba^{q-1}$
 we use the word $v=ba^{k-1}ba^{q-k-1}$. First we note, that the word $v$ fixes the state $0$.
The word $v$ moves all the other states in $C$ except $q-k$ along the cycle
in the same way as the word $ba^{q-1}$ does in lemma~\ref{Cycle_syn}.
The state $q-k$ leaves the cycle after applying the prefix $ba^{k-1}b$, but it can be easily seen that
$(q-k)\dt ba^{k-1}\boldsymbol{b}a^{q-k-1}=(q-k)\dt ba^{k-1}\boldsymbol{a}a^{q-k-1}.$
Thus we may treat the state $q-k$ as if it never left the cycle $C$. Following the argument
in lemma~\ref{Cycle_syn}, we conclude, that $C\dt v^{p-2}=\{0,q-r\}$, where $r$ is the remainder
of the division of $q$ by $p$. Finally, we observe that $0\dt ba^{k-1}ba^{q-p-k}=(q-r)\dt ba^{k-1}ba^{q-p-k}=q-p+1.$

Consider now the case $k=q-p.$ Let $w$ be a synchronizing word for the automaton $\mathrsfs{D}^{aa}(q,p,k)$ having minimal
possible length.
Since the incoming edges to the state $q-p+1$ have different labels, the word $w$ necessarily resets the automaton 
to the state $q-p+1+k.$
For convenience, let $t$ denote the state $q-p+1+k$.
Every word of length $k$ brings the state $q-p+1$ to the state $t$. Therefore, by lemma~\ref{rt} we have
$|w|\ge (p-1)(q-1)+k$. Suppose $|w|=(p-1)(q-1)+k+i$ for some $0\le i\le k-1.$ Consider the states $q-i$
(the state $0$, if $i=0$) and $q-p-i$. The prefix of $w$ of length $k+1+i$ will bring one of these states
to the state $t$ depending on the $(i+1)$st letter. The remaining $(p-1)(q-1)-1$ letters of $w$ will
move the state $t$ to itself. But this path is a combination of cycles of lengths
$p$ and $q$, which is impossible by theorem~\ref{coin_problem}. Consequently, $|w|\ge (p-1)(q-1)+2k=(p-1)(q-1)+2(q-p).$

Let us prove that the word $w_2=a^{q-p}(ba^{k-1}ba^{q-k-1})^{p-2}ba^{k-1}ba^{q-p}$ is synchronizing.
The prefix $a^{q-p}$ brings all the states lying outside the cycle $C=\{0,q-p+1,q-p+2,\ldots,q-1\}$ into $C$.
Arguing as in the previous case we conclude, that $C\dt (ba^{k-1}ba^{q-k-1})^{p-2}=\{0,q-r\}$.
It easy to see, that $0\dt ba^{k-1}ba^{q-p}=(q-r)\dt ba^{k-1}ba^{q-p}=t$.

%Then we assume, that $|w|=(p-1)(q-1)$. Thus, there is a directed path from the state $q-p$ to
%the state $q-p+1$ of length $(p-1)(q-1)$. Since both transitions from $q-p$ lead to $q-p+1$, we conclude,
%that there is a directed cycle on the state $q-p+1$ of length $(p-1)(q-1)-1,$ a contradiction.
%Thus, $|w|\ge(p-1)(q-1)+1$. Suppose, that $|w|\ge (p-1)(q-1)+i$ for $1\le i< q-p-k$, and let us prove, that
%$|w|\ge(p-1)(q-1)+i+1$. Indeed, if $|w|=(p-1)(q-1)+i$, then we consider the action of $w$ on the state $q-p-i$.
%Since $1 \leq q-p-i \leq q-p-1$ it follows that after applying the first $i+1$ letters of $w$ we end up in the state $q-p+1$.
%After applying the remaining $(p-1)(q-1)-1$ letters we should return to $q-p+1$, but this is impossible.
%Thus, $|w|\ge(p-1)(q-1)+i+1$ for $1\le i< q-p-k$. Therefore $\rt(\mathrsfs{A})\ge (p-1)(q-1)+q-p-k$.

%If $q<p$ then reset words are $ba^{q-p-k-1}(ba^{q-1})^{p-2}ba^{q-p}$
%and\\
%$a^{q-p-k}(ba^{k-1}ba^{q-k-1})^{p-2}ba^{k-1}ba^{q-p-k}$.
%
%If $q=p$ then reset word is $a^{q-p}(ba^{k-1}ba^{q-k-1})^{p-2}ba^{k-1}ba^{q-p}$.
\end{proof}
%\textbf{Acknowledgement}. The authors acknowledge support from the Presidential Programm for young researchers, grant MK-266.2012.1.

We can partially generalize this result as we did in theorem~\ref{wielandt_general} for the case of more than $q$ states.
%In this regard we need slight modification of lemma~\ref{congr}.
%\begin{lemma}
%\label{cong_increase}
%Given automaton $\mathrsfs{A}$ over alphabet $\Sigma$.
%If for any pair of states $s,t$ and arbitrary letter $\ell$ from $s \dt \ell = t \dt \ell$
%follows $s \sigma t$.
%Then $$\rt(\mathrsfs{A}/\sigma) = \rt(\mathrsfs{B}/\sigma)+1.$$
%\end{lemma}
%\begin{proof}
%Obvious
%\end{proof}
We consider a primitive digraph $D_\lambda (q,p,k)$ presented on Fig.~\ref{fig_B1_gen}, where $1\le\lambda<p$.
For convenience, we set $D_0(q,p,k)=D(q,p,k)$. 
Its colorings are denoted by $\mathrsfs{D}_\lambda^{aa}(q,p,k)$ and $\mathrsfs{D}_\lambda^{ab}(q,p,k)$.
\begin{figure}[ht]
 \begin{center}
  \unitlength=4pt
    \begin{picture}(45,45)(0,0)
    \gasset{Nw=5,Nh=5,Nmr=2.5}
    \thinlines
    \node(A6)(22.5,0){$t$}
\node(A0)(1.5,30){$0$}
\node(A1)(11.25,42){$1$}
\node[Nframe=n](A2)(33.75,42){$\ldots$}
\node(A3)(43.5,30){$k$}
\node[Nw=6](A4)(5.29,8){$q$-$1$}
\node(A5)(39.7,8){$s$}
\node[Nframe=n,NLdist=1](A7)(44,17){$\vdots$}
\node(A8)(10.5,27){$q$}
\node[Nframe=n](A9)(22.5,22.5){$\ddots$}
\node[Nw=7](A10)(33,15){$q$+$\lambda$-$1$}
\node[Nw=6](A11)(31.5,27){$q$+$\lambda$}
\node[Nframe=n,NLdist=1](A12)(22.5,16.5){$\vdots$}
\node[Nw=8](A13)(22.5,7.5){$q$+$2\lambda$-1}
\node[Nframe=n](A14)(12,2.6){$\ldots$}
\node[Nframe=n](A15)(33,2.6){$\ldots$}
    \drawedge[curvedepth=1](A0,A1){}
    \drawedge[curvedepth=2](A1,A2){}
    \drawedge[curvedepth=1](A2,A3){}
    \drawedge[curvedepth=3](A4,A0){}
    \drawedge[curvedepth=1](A3,A7){}
    \drawedge[curvedepth=1](A7,A5){}
    \drawedge[curvedepth=1](A5,A15){}
    \drawedge[curvedepth=1](A6,A14){}
    \drawedge[curvedepth=1](A14,A4){}
    \drawedge[curvedepth=1](A15,A6){}        
    \drawedge(A0,A8){}
    \drawedge(A8,A9){}    
    \drawedge(A9,A10){}    
    \drawedge(A10,A5){}
    \drawedge(A3,A11){}
    \drawedge[curvedepth=-1](A11,A12){}
    \drawedge(A12,A13){}    
    \drawedge(A13,A6){}    
%\node[Nframe=n](B6)(55,0){$\ldots$}
%\node[Nframe=n,NLdist=1](B7)(69,10){$\vdots$}
%\node(B0)(41,20){$0$}
%\node(B1)(47.5,28){$1$}
%\node[Nframe=n](B2)(62.5,28){$\ldots$}
%\node(B3)(69,20){$k$}
%\node[Nw=6](B4)(44.4,4.4){$q$-$1$}
%\node[Nw=7](B5)(65.6,4.4){$q$-$p$+$1$}
%    \drawedge[curvedepth=1](B0,B1){$b$}
%    \drawedge[curvedepth=2](B1,B2){$a,b$}
%    \drawedge[curvedepth=1](B2,B3){$a,b$}
%    \drawedge[curvedepth=2](B4,B0){$a,b$}
%    \drawedge[curvedepth=1](B3,B7){$b$}
%    \drawedge[curvedepth=1](B7,B5){}
%    \drawedge[curvedepth=1.5](B5,B6){$a,b$}
%    \drawedge[curvedepth=1](B6,B4){$a,b$}
%    \drawedge[curvedepth=2](B0,B5){$a$}
%    \drawedge[curvedepth=-2,ELpos=30,ELside=r](B3,B6){$a$}
%
%\node[Nframe=n](C6)(95,0){$\ldots$}
%\node[Nframe=n,NLdist=1](C7)(109,10){$\vdots$}
%\node(C0)(81,20){$0$}
%\node(C1)(87.5,28){$1$}
%\node[Nframe=n](C2)(102.5,28){$\ldots$}
%\node(C3)(109,20){$k$}
%\node[Nw=6](C4)(84.4,4.4){$q$-$1$}
%\node[Nw=7](C5)(105.6,4.4){$q$-$p$+$1$}
%    \drawedge[curvedepth=1](C0,C1){$b$}
%    \drawedge[curvedepth=2](C1,C2){$a,b$}
%    \drawedge[curvedepth=1](C2,C3){$a,b$}
%    \drawedge[curvedepth=2](C4,C0){$a,b$}
%    \drawedge[curvedepth=1](C3,C7){$a$}
%    \drawedge[curvedepth=1](C7,C5){}
%    \drawedge[curvedepth=1.5](C5,C6){$a,b$}
%    \drawedge[curvedepth=1](C6,C4){$a,b$}
%    \drawedge[curvedepth=2](C0,C5){$a$}
%    \drawedge[curvedepth=-2,ELpos=30,ELside=r](C3,C6){$b$}

    \end{picture}
 \end{center}
  \caption{The digraph $D_\lambda(q,p,k)$}
  \label{fig_B1_gen}
\end{figure}
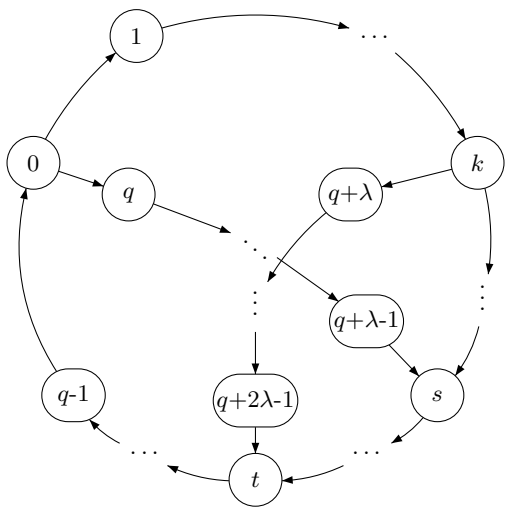
%We have the following
\begin{lemma}
\label{cong_B}
If $1\le\lambda< p$ and $z\in\{a,b\}$, then $\mathrsfs{D}_\lambda^{az}(q,p,k)/\sigma$ is equal to $\mathrsfs{D}_{\lambda-1}^{az}(q,p,k)$, and
$$ \rt(\mathrsfs{D}_\lambda^{az}(q,p,k)) = \rt(\mathrsfs{D}_{\lambda-1}^{az}(q,p,k)) + 1.$$
%Let $\mathrsfs{B}$ be an automaton whose underlying digraph is a $(q+2\lambda,q,p)$-graph with $1\le\lambda< p$. Then the underlying digraph of $\mathrsfs{B}/\sigma$ is a $(q+2(\lambda-1),q,p)$-graph, and $\rt(\mathrsfs{B})=\rt(\mathrsfs{B}/\sigma)+1.$
\end{lemma}
\begin{proof}
Let $w$ be a word synchronizing the automaton $\mathrsfs{D}_\lambda^{az}(q,p,k)$ having minimal length. Then $w$ resets the automaton  either to the state $s$, or to the state $t$. 
Let $x$ be the last letter of $w$, so that $w=w'x$. The word $w'$ brings the automaton $\mathrsfs{D}_\lambda^{az}(q,p,k)$ either to the set $\{q+\lambda-1,s-1\}$, or $\{q+2\lambda-1,t-1\}$. 
These two pairs of states form the two non-trivial $\sigma$-classes. Hence the factor automaton $\mathrsfs{D}_\lambda^{az}(q,p,k)/\sigma$ is equal to $\mathrsfs{D}_{\lambda-1}^{az}(q,p,k)$, and it is synchronized by $w'$. Thus $\rt(\mathrsfs{D}_\lambda^{az}(q,p,k)/\sigma)\le \rt(\mathrsfs{D}_\lambda^{az}(q,p,k))-1$.
On the other hand, by lemma~\ref{congr} we have $\rt(\mathrsfs{D}_\lambda^{az}(q,p,k)/\sigma)\ge \rt(\mathrsfs{D}_\lambda^{az}(q,p,k))-1$, and we get the required equality. 
\end{proof}
\begin{theorem}
\label{B_general}
If $1\le\lambda< p$, then\\ 
$(i)$ $\rt(\mathrsfs{D}_\lambda^{ab}(q,p,k))=(p-1)(q-1)+q-p-k+\lambda;$\\
$(ii)$ $\rt(\mathrsfs{D}_\lambda^{aa}(q,p,k))=(p-1)(q-1)+q-p-k+\lambda,$ if $k<q-p$;\\ 
$(iii)$ $\rt(\mathrsfs{D}_\lambda^{aa}(q,p,k))=(p-1)(q-1)+2(q-p)+\lambda,$ if $k=q-p$.
\end{theorem}
\begin{proof}
Since there are $\lambda$ states both on the path from the state $0$ to $s$, and from $k$ to $t$, and $k\le k-p$, lemma~\ref{cong_B} can be applied $\lambda$ times. Each time lemma~\ref{cong_B} is applied, the reset threshold is decreased strictly by one.
In the end, from the automaton $\mathrsfs{D}_\lambda^{ab}(q,p,k)$ we obtain the automaton $\mathrsfs{D}_0^{aa}(q,p,k)$, whose reset
threshold is known by theorem~\ref{B2_rt}. Therefore, we have 
$rt(\mathrsfs{D}_\lambda^{ab}(q,p,k))=(p-1)(q-1)+q-p-k+\lambda$. In an analogous way from the automaton $\mathrsfs{D}_\lambda^{aa}(q,p,k)$ we obtain the automaton $\mathrsfs{D}_0^{aa}(q,p,k)$.
Applying theorem~\ref{B1_rt}, we obtain $\rt(\mathrsfs{D}_\lambda^{aa}(q,p,k))=(p-1)(q-1)+q-p-k+\lambda$ in case $k<q-p$, and $\rt(\mathrsfs{D}_\lambda^{aa}(q,p,k))=(p-1)(q-1)+2(q-p)+\lambda$ if $k=q-p$.
\end{proof}

%Here goes definition of automata $\mathrsfs{B}_1$ and $\mathrsfs{B}_2$ with equal number of letters
%on edges and reducible to them.

%\begin{theorem}
%Reset threshold of subdivided examples is equal to.
%\end{theorem}
%\begin{proof}
%We use theorems and lemma~\ref{cong_increase}.
%\end{proof}

The case of non-equal number of
states on the paths from 0 to $s$ and from $k$ to $t$ is much more technical, and will be published elsewhere.

%\section{Conclusion}

\end{document}